\documentclass[a4]{article}

\usepackage{authblk}
\usepackage{amsmath}
\usepackage{amssymb}
\usepackage{amsthm}
\usepackage{a4wide}
\usepackage[mathscr]{eucal}
\usepackage{mathtools}
\usepackage{bm}
\usepackage{bbm}
\usepackage{enumitem}
\usepackage{hyperref}
\usepackage{booktabs}
\usepackage[linesnumbered,ruled]{algorithm2e}
\usepackage[nocompress]{cite}

\usepackage{tikz}
\usepackage{graphicx}
\usepackage{subcaption}

\newtheorem{theorem}{Theorem}

\newtheorem{corollary}[theorem]{Corollary}

\newtheorem{lemma}[theorem]{Lemma}

\newcommand{\prob}{\mathbb{P}}
\newcommand{\Prob}[1]{\prob\left(#1\right)}

\newcommand{\expec}{\mathbb{E}}
\newcommand{\Exp}[1]{\expec\left[#1\right]}

\def\t{\tau}

\def\P{\mathbb{P}}
\def\T{\Theta}

\newcommand\abs[1]{\left|#1\right|}
\newcommand{\me}{\textup{e}}

\allowdisplaybreaks

\setenumerate{label={\normalfont{(\roman*)}}} 
\SetLabelAlign{LeftAlignWithIndent}{\hspace*{1.0ex}\makebox[1.25em][l]{#1}}
\numberwithin{equation}{section}

\title{Finding induced subgraphs in scale-free inhomogeneous random graphs}
\author[1]{Ellen Cardinaels}
\author[2]{Johan S.H. van Leeuwaarden}
\author[3]{Clara Stegehuis}
\affil[1]{Eindhoven University of Technology}
\affil[2]{Tilburg University}
\affil[3]{Twente University}

\graphicspath{{Figures/}}

\begin{document}
		\maketitle
	\begin{abstract}
		We study the problem of finding a copy of a specific induced subgraph on inhomogeneous random graphs with infinite variance power-law degrees. 
We provide a fast algorithm that finds a copy of any connected graph $H$ on a fixed number of $k$ vertices as an induced subgraph in a random graph with $n$ vertices. By exploiting the scale-free graph structure, the algorithm runs in $O(n k)$ time for small values of $k$. As a corollary, this shows that the induced subgraph isomorphism problem can be solved in time $O(nk)$ for the inhomogeneous random graph. We test our algorithm on several real-world data sets.
	\end{abstract}

\section{Introduction}

The induced subgraph isomorphism problem asks whether a large graph $G$ contains a connected graph $H$ as an induced subgraph. When $k$ is allowed to grow with the graph size $n$, this problem is NP-hard in general. For example, $k$-clique and $k$ induced cycle, special cases of $H$, are known to be NP-hard~\cite{karp1972,garey2011}. For fixed $k$, this problem can be solved in polynomial time $O(n^k)$ by searching for $H$ on all possible combinations of $k$ vertices. Several randomized and non-randomized algorithms exist to improve upon this trivial way of finding $H$~\cite{williams2015,grochow2007,schreiber2005,omidi2009}.

A second problem we investigate is how to find a subgraph, when we know it exists, more efficiently than the trivial O(nk) algorithm. 

On real-world networks, many algorithms were observed to run much faster  than predicted by the worst-case running time of algorithms. This may be ascribed to some of the properties that many real-world networks share~\cite{brach2016}, such as the power-law degree distribution found in many networks~\cite{albert1999,faloutsos1999,jeong2000,vazquez2002}. One way of exploiting these power-law degree distributions is to design algorithms that work well on random graphs with power-law degree distributions.
For example, finding the largest clique in a network is NP-complete for general networks~\cite{karp1972}. However, in random graph models such as the Erd\H{o}s-R\'enyi random graph and the inhomogeneous random graph, their specific structures can be exploited to design {fixed parameter tractable} (FPT) algorithms that efficiently find a clique of size $k$~\cite{fountoulakis2015,friedrich2015a} or the largest independent set~\cite{heydari2017}.

In this paper, we study algorithms that are designed to efficiently find subgraphs in the inhomogeneous random graph, a random graph model that can generate graphs with a power-law degree distribution~\cite{boguna2003,park2004,bollobas2007,britton2006,norros2006,chung2002}. 
The inhomogeneous random graph has a densely connected core containing many cliques, consisting of vertices with degrees $\sqrt{n\log(n)}$ and larger. In this densely connected core, the probability of an edge being present is close to one, so that it contains many complete graphs~\cite{janson2010}. This observation was exploited in~\cite{friedrich2015} to efficiently determine whether a clique of size $k$ occurs as a subgraph in an inhomogeneous random graph.  
When searching for {\it induced} subgraphs however, some edges are required not to be present. Therefore, searching for induced subgraphs in the entire core is not efficient. We show that a connected subgraph $H$ can be found as an induced subgraph by scanning only vertices that are on the boundary of the core: vertices with degrees proportional to $\sqrt{n}$. 

We present an algorithm that first selects the set of vertices with degrees proportional to $\sqrt{n}$, and then randomly searches for $H$ as an induced subgraph on a subset of $k$ of those vertices. The first algorithm we present does not depend on the specific structure of $H$. For general sparse graphs, the best known algorithms to solve subgraph isomorphism on 3 or 4 vertices run in $O(n^{1.41})$ or $O(n^{1.51})$ time with high probability~\cite{williams2015}. For small values of $k$, our algorithm finds the desired subgraph on $k$ nodes in linear time with high probability on inhomogeneous random graphs. However, the graph size needs to be very large for our algorithm to perform well. We therefore present a second algorithm that again selects the vertices with degrees proportional to $\sqrt{n}$, and then searches for induced subgraph $H$ in a more efficient way. This algorithm has the same performance guarantee as our first algorithm, but performs much better in simulations. 

We test our algorithm on large inhomogeneous random graphs, where it indeed efficiently finds induced subgraphs. We also test our algorithm on real-world network data with power-law degrees. There our algorithm does not perform well, probably due to the fact that the densely connected core of some real-world networks may not be the vertices of degrees at least proportional to $\sqrt{n}$. We then show that a slight modification of our algorithm that looks for induced subgraphs on vertices of degrees proportional to $n^{\gamma}$ for some other value of $\gamma$ performs better on real-world networks, where the value of $\gamma$ depends on the specific network.

\paragraph{Notation.}
 We say that a sequence of events $(\mathcal{E}_n)_{n\geq 1}$ happens with high probability (w.h.p.) if $\lim_{n\to\infty}\Prob{\mathcal{E}_n}=1$. Furthermore, we write $f(n)=o(g(n))$ if $\lim_{n\to\infty}f(n)/g(n)=0$, and $f(n)=O(g(n))$ if $|f(n)|/g(n)$ is uniformly bounded, where $(g(n))_{n\geq 1}$ is nonnegative. Similarly, if $\limsup_{n\to\infty}\abs{f(n)}/g(n)>0$, we say that $f(n)=\Omega(g(n))$ for nonnegative $(g(n))_{n\geq 1}$. We write $f(n)=\Theta(g(n))$ if $f(n)=O(g(n) )$ as well as $f(n)=\Omega(g(n))$. 

\subsection{Model}
As a random graph null model, we use the inhomogeneous random graph or hidden variable model~\cite{boguna2003,park2004,bollobas2007,britton2006,norros2006,chung2002}. Every vertex is equipped with a weight $w$.
We assume that the weights are i.i.d.\ samples from the power-law distribution
\begin{equation}\label{eq:pl}
\Prob{w>k}=Ck^{1-\tau}
\end{equation}
for some constant $C$ and for $\tau\in(2,3)$. Two vertices with weights $w$ and $w'$ are connected with probability
\begin{equation}
p(w,w') = \min\left(\frac{w w'}{\mu n},1\right),
\end{equation}
where $\mu$ denotes the mean value of the power-law distribution~\eqref{eq:pl}. Choosing the connection probability in this way ensures that the expected degree of a vertex with weight $w$ is $w$. We denote the degree of vertex $i$ in the inhomogeneous random graph by $D_i$ and its weight by $w_i$.

\subsection{Algorithms}
We now describe two randomized algorithms that determine whether a connected graph $H$ is an induced subgraph in an inhomogeneous random graph and finds the location of such a subgraph if it exists. Algorithm~\ref{alg:motif} selects the vertices in the inhomogeneous random graph that are on the boundary of the core of the graph: vertices with degrees slightly below $\sqrt{\mu n}$. Then, the algorithm randomly divides these vertices into sets of $k$ vertices. If one of these sets contains $H$ as an induced subgraph, the algorithm terminates and returns the location of $H$. If this is not the case, then the algorithm fails. In the next section, we show that for $k$ small enough, the probability that the algorithm fails is small. This means that $H$ is present as an induced subgraph on vertices that are on the boundary of the core with high probability.

Algorithm~\ref{alg:motif} is similar to the algorithm in~\cite{friedrich2015a} designed to find cliques in random graphs. The major difference is that the algorithm to find cliques looks for cliques on all vertices with degrees larger than $\sqrt{f_1\mu n}$ for some function $f_1$. This algorithm is not efficient for detecting subgraphs other than cliques, since vertices with high degrees will be connected with probability close to one. 

\begin{algorithm}
	\caption{Finding induced subgraph $H$ (random search)}
	\label{alg:motif}
	\SetKwInOut{Input}{Input}\SetKwInOut{Output}{Output}
	\Input{$H=(V_H,E_H)$, $G=(V_G,E_G)$, $\mu$, $f_1=f_1(n)$, $f_2=f_2(n)$.}
	\Output{Location of $H$ in $G$ or fail.}
	Define $n=\abs{V}$, $I_n=[\sqrt{f_1 \mu n}, \sqrt{f_2\mu n}]$, set $k=|V_H|$ and $V'=\emptyset$.\\
	\For{$i\in V$}{
	\lIf{ $D_i\in I_n$}{ $V'=V'\cup i$}}
	Divide the vertices in $V'$ randomly into $\lfloor \abs{V'}/k\rfloor$ sets $S_1,\dots,S_{\lfloor\abs{V'}/k\rfloor}$.\\
	\For{$j=1,\dots,\lfloor\abs{V'}/k\rfloor$}{
	\lIf{ $H$ is an induced subgraph on $S_j$}{\Return location of $H$}}
\end{algorithm}

The following theorem gives a bound for the performance of Algorithm~\ref{alg:motif} for small values of $k$.
\begin{theorem}\label{thm:graphletfind}
	Choose $f_1=f_1(n)\geq 1/\log(n)$ and $f_2(n)$ such that for some $a<1, b<1$, $f_1<af_2$ and $f_2<b<1$ for all $n$. Let $k<\log^{1/3}(n)$. 
	Then, with high probability, Algorithm~\ref{alg:motif} detects induced subgraph $H$ on $k$ vertices in an inhomogeneous random graph with $n$ vertices and weights distributed as in~\eqref{eq:pl} in time $O(n k)$.
\end{theorem}
Thus, for constant values of $k$, Algorithm~\ref{alg:motif} finds an instance of $H$ in linear time. 

A problem with parameter $k$ is called fixed parameter tractable (FPT) if it can be solved in $f(k)n^{O(1)}$ time for some function $f(k)$, and it is called typical FPT (typFPT) if it can be solved in $f(k)n^{O(1)}$ with high probability~\cite{fountoulakis2014}. As a corollary of Theorem~\ref{thm:graphletfind} we obtain that the induced subgraph problem on the inhomogeneous random graph is in typFPT for any subgraph $H$, similarly to the $k$-clique problem on inhomogeneous random graphs~\cite{friedrich2015a}.
\begin{corollary}\label{cor:fpt}
	The induced subgraph problem on the inhomogeneous random graph is in typFPT. 
	\end{corollary}

In theory Algorithm~\ref{alg:motif} detects any motif on $k$ vertices in linear time for small $k$. However, this only holds for large values of $n$, which can be understood as follows. In Lemma~\ref{lem:poly}, we show that $\abs{V'}=\Theta(n^{(3-\tau)/2})$, thus tending to infinity as $n$ grows large. However, when $n=10^7$ and $\tau=2.5$, this means that the size of the set $V'$ is only proportional to $10^{1.75}=56$ vertices. Therefore, the number of sets $S_j$ constructed in Algorithm~\ref{alg:motif} is also small. Even though the probability of finding motif $H$ in any such set is proportional to a constant, this constant may be small, so that for finite $n$ the algorithm almost always fails. Thus, for Algorithm~\ref{alg:motif} to work, $n$ needs to be large enough so that $n^{(3-\tau)/2}$ is large as well.

The algorithm can be significantly improved by changing the search for $H$ on vertices in set $V'$. In Algorithm~\ref{alg:motif2} we propose a search for motif $H$ similar to the Kashtan motif sampling algorithm~\cite{kashtan2004}. Rather than sampling $k$ vertices randomly, it samples one vertex randomly, and then randomly increases the set $S$ by adding vertices in its neighborhood. This already guarantees the vertices in list $S_j$ to be connected, making it more likely for them to form a specific connected motif together. In particular, we expand the list $S_j$ in such a way that the vertices in $S_j$ are guaranteed to form a spanning tree of $H$ as a subgraph. This is ensured by choosing the list $T^H$ that specifies at which vertex in $S_j$ we expand $S_j$ by adding a new vertex. For example, if $k=4$ and we set $T^H=[1,2,3]$ we first add an edge to the first vertex, then we look for a random neighbor of the previously added vertex, and then we add a random neighbor of the third added vertex. Thus, setting $T^H=[1,2,3]$ ensures that the set $S_j$ contains a path of length three, whereas setting $T^H=[1,1,1]$ ensures that the set $S_j$ contains a star-shaped subgraph. Depending on which subgraph $H$ we are looking for, we can define $T^H$ in such a way that we ensure that the set $S_j$ at least contains a spanning tree of motif $H$ in Step 6 of the algorithm. 

The selection on the degrees ensures that the degrees are sufficiently high so that probability of finding such a connected set on $k$ vertices is high, as well as that the degrees are sufficiently low to ensure that we do not only find complete graphs because of the densely connected core of the inhomogeneous random graph. 
The probability that Algorithm~\ref{alg:motif2} indeed finds the desired motif $H$ in any check is of constant order of magnitude, similar to Algorithm~\ref{alg:motif}. Therefore, the performance guarantee of both algorithms is similar. However, for several synthetic and real-world data sets, we show in Section~\ref{sec:exp} that Algorithm~\ref{alg:motif2} performs much better, since for finite $n$, $k$ connected vertices are more likely to form a motif than $k$ randomly chosen vertices.
 
\begin{algorithm}
	\caption{Finding induced subgraph $H$ (neighborhood search)}
	\label{alg:motif2}
	\SetKwInOut{Input}{Input}\SetKwInOut{Output}{Output}
	\Input{$H$, $G=(V,E)$, $\mu$, $f_1=f_1(n)$, $f_2=f_2(n)$, $s$.}
	\Output{Location of $H$ in $G$ or fail.}
	Define $n=\abs{V}$, $I_n=[\sqrt{f_1 \mu n},\sqrt{f_2\mu n}]$ and set $V'=\emptyset$.\\
	\For{$i\in V$}{
		\lIf{ $D_i\in I_n$}{ $V'=V'\cup i$}}
	Let $G'$ be the induced subgraph of $G$ on vertices $V'$. \\
	Set $T^H$ consistently with motif $H$ .\\
	\For{j=1,\dots,s}{
		Pick a random vertex $v\in V'$ and set $S_j=\{v\}$.\\
		\While{$\abs{S_j}\neq k$}{Pick a random $v'\in N_{G'}(S_j[T^H[j]]):v'\notin S_j$\\
		 Add $v'$ to $S_j$.}
		\lIf{ $H$ is an induced subgraph on $S_j$}{\Return location of $H$}}
\end{algorithm}

The following theorem shows that indeed Algorithm~\ref{alg:motif2} has similar performance guarantees as Algorithm~\ref{alg:motif}.
\begin{theorem}\label{thm:alg2}
	Choose $f_1=f_1(n)\geq 1/\log(n)$ and $f_1<f_2<1$. Choose $s=\Omega(n^{\alpha})$ for some $0<\alpha<1$, such that $s\leq n/k$. Then, Algorithm~\ref{alg:motif2} detects induced subgraph $H$ on $k<\log^{1/3}(n)$ vertices on an inhomogeneous random graph with $n$ vertices and weights distributed as in~\eqref{eq:pl} in time $O(nk)$ with high probability.
\end{theorem}

{The proofs of Theorem~\ref{thm:graphletfind} and~\ref{thm:alg2} rely on the fact that for small $k$, any subgraph on $k$ vertices is present in $G'$ with high probability. This means that after the degree selection step of Algorithms~\ref{alg:motif} and~\ref{alg:motif2}, for small $k$, any motif finding algorithm can be used to find motif $H$ on the remaining graph $G'$, such as the Grochow-Kellis algorithm~\cite{grochow2007}, the MAvisto algorithm~\cite{schreiber2005} or the MODA algorithm~\cite{omidi2009}. In the proofs of Theorem~\ref{thm:graphletfind} and~\ref{thm:alg2}, we show that $G'$ has $\Theta(n^{(3-\tau)/2})$ vertices with high probability. Thus, the degree selection step reduces the problem of finding a motif $H$ on $n$ vertices to finding a motif on a graph with $\Theta(n^{(3-\tau)/2})$ vertices, significantly reducing the running time of the algorithms. }



\section{Proof of Theorems~\ref{thm:graphletfind} and~\ref{thm:alg2}}
We prove Theorem~\ref{thm:graphletfind} using two lemmas. 
The first lemma relates the degrees of the vertices to their weights.
The connection probabilities in the inhomogeneous random graph depend on the weights of the vertices. In Algorithm~\ref{alg:motif}, we select vertices based on their degrees instead of their unknown weights. The following lemma shows that the weights of the vertices in $V'$ are close to their degrees.
\begin{lemma}[Degrees and weights] \label{lem:degweight}
	Fix $\varepsilon>0$, and define $J_n=[(1-\varepsilon)\sqrt{f_1 \mu n},(1+\varepsilon)\sqrt{f_2\mu n}]$. Then, for some $K>0$,
	\begin{equation}
	\Prob{\exists i\in V': w_i \notin J_n}\leq K n\exp\Big(-\varepsilon^2\sqrt{\mu n}\min\Big(\frac{\sqrt{f_1}}{1-\varepsilon},\frac{\sqrt{f_2}}{1+\varepsilon}\Big)/2\Big).
	\end{equation}
\end{lemma}
\begin{proof}
	Fix a vertex $i\in V$. 
	Then,
	\begin{equation}\label{eq:Pwismall}
	\begin{aligned}[b]
	\Prob{w_i<(1-\varepsilon)\sqrt{f_1 \mu n}, \  D_i\in I_n}&=\frac{\Prob{D_i\in I_n\mid w_i< (1-\varepsilon)\sqrt{f_1 \mu n}}}{\Prob{w_i<(1-\varepsilon)\sqrt{f_1\mu n}}}\\
	&\leq \frac{\Prob{D_i>\sqrt{f_1 \mu n}\mid w_i=(1-\varepsilon)\sqrt{f_1\mu n}}}{1-C((1-\varepsilon)\sqrt{f_1\mu n})^{1-\tau}}\\
	&\leq K_1\Prob{D_i>\sqrt{f_1 \mu n}\mid w_i=(1-\varepsilon)\sqrt{f_1\mu n}},
	\end{aligned}
	\end{equation} 
	for some $K_1>0$. 
	
	Here the first inequality follows because the probability that a vertex with weight $w_1$ has degree at least $\sqrt{f_1 \mu n}$ is larger than the probability that a vertex of weight $w_2$ has degree at least $\sqrt{f_1 \mu n}$ when $w_1>w_2$.
	Conditionally on the weights, $D_i$ is the sum of $n-1$ independent indicators indicating the presence of an edge between vertex $i$ and the other vertices and that $\Exp{D_i}=w_i$. Therefore, by the Chernoff bound
	\begin{equation}
	\Prob{D_i>w_i(1+\delta)}\leq \exp\big(-\delta^2w_i/2\big).
	\end{equation}
	Therefore, choosing $\delta=\varepsilon/(1-\varepsilon)$ yields 
	\begin{equation}\label{eq:Pdlarge}
	\begin{aligned}[b]
	\Prob{D_i>\sqrt{f_1 \mu n}\mid w_i=(1-\varepsilon)\sqrt{f_1 \mu n}}&\leq \exp\bigg(-\frac{\varepsilon^2\sqrt{f_1\mu n}}{2(1-\varepsilon)}\bigg)(1+o(1)).
	\end{aligned}
	\end{equation}
	Combining this with~\eqref{eq:Pwismall} and taking the union bound over all vertices then results in
	\begin{equation}
	\Prob{\exists i: D_i\in I_n,w_i<(1-\varepsilon)\sqrt{f_1\mu n}}\leq K_2 n \exp\Big(-\frac{\varepsilon^2}{2(1-\varepsilon)}\sqrt{f_1\mu n}\Big),
	\end{equation}
	for some $K_2>0$. Similarly,
	\begin{equation}
	\Prob{\exists i: D_i\in I_n,w_i>(1+\varepsilon)\sqrt{f_2\mu n}}\leq K_3 n \exp\Big(-\frac{\varepsilon^2}{2(1+\varepsilon)}\sqrt{f_2\mu n}\Big),
	\end{equation}
	for some $K_3>0$, which proves the lemma.
\end{proof}

\begin{lemma}[Weights and degrees]\label{lem:weightslb}
	Fix $\varepsilon>0$, sufficiently small so that $\tilde{J}_n=[(1+\varepsilon)\sqrt{f_1 \mu n},(1-\varepsilon)\sqrt{f_2\mu n}]$ is a non-empty interval. Then, for some $K>0$,
	\begin{equation}
	\Prob{\exists i: w_i \in \tilde{J}_n, i\notin V'}\leq K n\exp\Big(-\varepsilon^2\sqrt{\mu n}\min\Big(\frac{\sqrt{f_1}}{1-\varepsilon},\frac{\sqrt{f_2}}{1+\varepsilon}\Big)/2\Big).
	\end{equation}
\end{lemma}
\begin{proof}
	Fix a vertex $i$ with $w_i\in\tilde{J}_n$. Then,
	\begin{equation}
		\Prob{D_i<\sqrt{f_1 \mu n}\mid w_i\in\tilde{J}_n}
		\leq \Prob{D_i<\sqrt{f_1 \mu n}\mid w_i=(1+\varepsilon)\sqrt{f_1\mu n}}.
	\end{equation}
	Similarly to~\eqref{eq:Pdlarge},
	\begin{equation}
		 \Prob{D_i<\sqrt{f_1 \mu n}\mid w_i=(1+\varepsilon)\sqrt{f_1\mu n}} \leq \exp\Big(-\frac{\varepsilon^2\sqrt{f_1 n}}{2(1+\varepsilon)}\Big),
	\end{equation}
	so that 
	\begin{equation}
	\Prob{\exists i: w_i \in \tilde{J}_n, D_i<\sqrt{f_1 \mu n}}\leq K_1n\exp\Big(-\frac{\varepsilon^2\sqrt{f_1 n}}{2(1+\varepsilon)}\Big).
	\end{equation}
	Similarly,
		\begin{equation}
	\Prob{\exists i: w_i \in \tilde{J}_n, D_i>\sqrt{f_2 \mu n}}\leq K_2n\exp\Big(-\frac{\varepsilon^2\sqrt{f_2 n}}{2(1-\varepsilon)}\Big).
	\end{equation}
\end{proof}

The second lemma shows that after deleting all vertices with degrees outside of $I_n$ defined in Step 1 of Algorithm~\ref{alg:motif}, still polynomially many vertices remain with high probability.
\begin{lemma}{Polynomially many nodes remain.}\label{lem:poly}
	There exists $\gamma, \gamma'>0$ such that
	\begin{equation}
	\Prob{|V'|<\gamma n^{(3-\tau)/2}}\leq 2\exp\left(-\Theta(n^{(3-\tau)/2})\right)
	\end{equation}
	and
	\begin{equation}
	\Prob{|V'|>\gamma' n^{(3-\tau)/2}\log^{\tau-1}(n)}\leq 2\exp\left(-\Theta(n^{(3-\tau)/2})\right)
	\end{equation}
\end{lemma}
\begin{proof}
	Let $\mathcal{E}$ denote the event that all vertices with $w_i\in\tilde{J}_n$ satisfy $i\in V'$ for some $\varepsilon>0$, with $\tilde{J}_n$ as in Lemma~\ref{lem:weightslb}. Let $W'$ be the set of all vertices with weights in $\tilde{J}_n$. Conditioned on the event $\mathcal{E}$, any vertex in $W'$ is also in $V'$ so that $|W'|\leq |V'|$.
	Then, by Lemma~\ref{lem:degweight}
	\begin{align}\label{eq:VWl}
		\Prob{|V'|<\gamma n^{(3-\tau)/2}}& \leq \Prob{|W'|<\gamma n^{(3-\tau)/2}}+
		K n\exp\Big(-\varepsilon^2\sqrt{\mu n}\min\Big(\frac{\sqrt{f_1}}{1-\varepsilon},\frac{\sqrt{f_2}}{1+\varepsilon}\Big)/2\Big)
	\end{align}
Furthermore,
	\begin{equation}
	\begin{aligned}[b]
	\Prob{w_i\in \tilde{J}_n}  & = C((1+\varepsilon)\sqrt{f_1\mu n})^{1-\tau}- C((1-\varepsilon)\sqrt{f_2\mu n})^{1-\tau} \geq c_1(\sqrt{n})^{1-\tau}
	\end{aligned}
	\end{equation}
	for some constant $c_1>0$ when $\varepsilon$ is sufficiently small, because $f_1<af_2$ for some constant $a<1$ by assumption. Thus, each of the $n$ vertices is in set $W'$ independently with probability at least $c_1(\sqrt{\mu n})^{1-\tau}$. Choose $0<\gamma<c_1$. Applying the multiplicative Chernoff bound then shows that
	\begin{equation}\label{eq:Wbound}
	\begin{aligned}[b]
	\Prob{|W'|<\gamma n^{(3-\tau)/2}}& \leq \exp\left(-\frac{(c_1-\gamma)^2}{2c_1} n^{(3-\tau)/2}\right),
	\end{aligned}
	\end{equation}
	which proves the first part of the lemma together with~\eqref{eq:VW}.
	
	Let $\mathcal{E}'$ denote the event that all vertices $i\in V'$ satisfy $w_i\in J_n$ for some $\varepsilon>0$, with ${J}_n$ as in Lemma~\ref{lem:degweight}. Let $U'$ be the set of all vertices with weights in ${J}_n$. On the event $\mathcal{E}'$, any vertex in $V'$ is also in $U'$ so that $|U'|\geq |V'|$.
	Then, by Lemma~\ref{lem:degweight}
	\begin{align}\label{eq:VW}
	\Prob{|V'|>\gamma' n^{(3-\tau)/2}}& \leq \Prob{|U'|>\gamma' n^{(3-\tau)/2}}+
	K n\exp\Big(-\varepsilon^2\sqrt{\mu n}\min\Big(\frac{\sqrt{f_1}}{1-\varepsilon},\frac{\sqrt{f_2}}{1+\varepsilon}\Big)/2\Big).
	\end{align}
	Furthermore, for some $c_2>0$,
	\begin{equation}
	\Prob{w_i\in J_n}   = C((1-\varepsilon)\sqrt{f_1\mu n})^{1-\tau}- C((1+\varepsilon)\sqrt{f_2\mu n})^{1-\tau} \leq c_2(\sqrt{n}/\log(n))^{1-\tau},
	\end{equation}
	where we used that $f_1\geq 1/\log(n)$.
	Similarly to~\eqref{eq:Wbound}, 
	\begin{equation}
	\begin{aligned}[b]
	\Prob{|U'|>\gamma' n^{(3-\tau)/2}\log^{\tau-1}(n)}& \leq \exp\left(-\frac{(c_2-\gamma')^2}{2c_2}\log^{\tau-1}(n) n^{(3-\tau)/2}\right),
	\end{aligned}
	\end{equation}
	which proves the second part of lemma.
\end{proof}

We now use these lemmas to prove Theorem~\ref{thm:graphletfind}.\\\\
\noindent
{\it Proof of Theorem~\ref{thm:graphletfind}.}
	We condition on the event that $V'$ is of polynomial size (Lemma~\ref{lem:poly}) and that the weights are within the constructed lower and upper bounds (Lemma~\ref{lem:degweight}), since both events occur with high probability. This bounds the edge probability between any pair of nodes $i$ and $j$ in $V'$ as
	\begin{equation}\label{eq:pijub}
	p_{ij} <\min\left(\frac{(1+\varepsilon)\sqrt{f_2\mu n}(1+\varepsilon)\sqrt{f_2\mu n}}{\mu n},1\right) = f_2(1+\varepsilon)^2.
	\end{equation}
	Because $f_2<b<1$ for some constant $b$, $p_{ij}\le p_+ <1$ for some constant $p_+$ if we choose $\varepsilon$ small enough. Similarly,
	\begin{equation} \label{eq:pijlb}
	p_{ij} >\min\left(\frac{(1-\varepsilon)^2\sqrt{f_1\mu n}^2}{\mu n},1\right)= \T\left(\frac{1}{\log(n)}\right),
	\end{equation}
	by our choice of $f_1$, so that $p_{ij} \ge  c_2/\log(n)=:p_-$ for some constant $c_2$.
	Let $E:=|E_H|$ be the number of edges in $H$. We upper bound the probability of not finding $H$ in one of the partitions of size $k$ of $V'$ as $1-p_-^{E}(1-p_+)^{\binom{k}{2}-E}$. Since all partitions are disjoint we can upper bound the probability of not finding $H$ in any of the partitions as
	\begin{equation}
	\P\left( H\text{ not in any set of partitions} \right) \le \left(1-p_-^E(1-p_+)^{\binom{k}{2}-E}\right)^{\left \lfloor{\frac{|V'|}{k}}\right \rfloor}.
	\end{equation}
	Using that $E\le k^2$, $\binom{k}{2}-E \le k^2$ and that $1-x\leq \me^{-x}$ results in
	\begin{equation}\label{eq:pHbound}
	\P\left( H\text{ not in the partitions} \right) \le \exp\left( -p_-^{k^2}(1-p_+)^{k^2} \left \lfloor{\frac{|V'|}{k}}\right \rfloor  \right).
	\end{equation}
	Since $\abs{V'} = \Omega\left(n^{\frac{3-\t}{2}}\right)$, 
	$\left \lceil {|V'|/k} \right \rceil \ge d n^{\frac{3-\t}{2}}/k$ for some constant $d>0$. We fill in the expressions for $p_-$ and $p_+$, with $c_3>0$ a constant
	\begin{equation} \label{temp}
	\P\left( H\text{ not in the partitions} \right) \le \exp \left(- \frac{d n^{\frac{3-\t}{2}}}{k} \left( \frac{c_3}{\log n}\right)^{k^2}  \right).
	\end{equation}
	Now apply that $k\le \log^{\frac{1}{3}} (n)$. Then 
	\begin{equation}
	\begin{array}{rcl}
	\P\left( H\text{ not in the partitions} \right) &\le &
	\exp \left(- \frac{d n^{\frac{3-\t}{2}}}{\log^{\frac{1}{3}} n} \left( \frac{c_3}{\log n}\right)^{\log^{\frac{2}{3}} n}   \right)\\
	&\le & \exp \left(- d n^{\frac{3-\t}{2}-o(1)}\right).\\
	\end{array}
	\end{equation}
	Hence, the inner expression grows polynomially such that the probability of not finding $H$ in one of the partitions is negligibly small. The running time of the partial search is given by
	\begin{equation}\label{eq:psearch}
		\frac{\abs{V'}}{k}{k\choose 2}\leq \frac{n}{k}{k\choose 2}\leq n k\leq n \me^{k^4},
	\end{equation}
	which concludes the proof for $k\leq \log^{1/3}(n)$.
	\hfill \qed
\vspace{0.2cm}
\noindent
\begin{proof}[Proof of Corollary~\ref{cor:fpt}.]
	If $k>\log^{\frac{1}{3}} (n)$, $n<e^{k^3}$, so that the time it takes to solve the subgraph isomorphism problem is bounded by a function of $k$.
	
	For $k\le\log^{\frac{1}{3}}( n)$, Theorem~\ref{thm:graphletfind} shows that the induced subgraph isomorphism problem can be solved in time $nk\leq n \me^{k^4}$. Thus, with high probability the induced subgraph isomorphism problem can be solved in $n\me^{k^4}$ time, which proves that it is in typFPT.
\end{proof}

\begin{proof}[Proof of Theorem~\ref{thm:alg2}.]
The proof of Theorem~\ref{thm:alg2} is very similar to the proof of Theorem~\ref{thm:graphletfind}. The only way Algorithm~\ref{alg:motif2} differs from Algorithm~\ref{alg:motif} is in the selection of the sets $S_j$. As in the previous theorem, we condition on the event that $\gamma n^{(3-\tau)/2}\leq \abs{V'}\leq \gamma'n^{(3-\tau)/2}\log^{\tau-1}(n)$ (Lemma~\ref{lem:poly}) and that the weights of the vertices in $G'$ are bounded as in Lemma~\ref{lem:degweight}. 

The graph $G'=(V',E')$ constructed in Step 5 of Algorithm~\ref{alg:motif2} then consists of $\Theta(n^{(3-\tau)/2})$ vertices. Furthermore, by the bound~\eqref{eq:pijlb} on the connection probabilities of all vertices in $G'$, the expected degree of a vertex $i$ in $G'$, $D_{i,G'}$, satisfies $\Exp{D_{i,G'}}=\Omega(n^{(3-\tau)/2}/\log(n))$. We can use similar arguments as in Lemma~\ref{lem:degweight} to show that $D_{i,G'}=\Omega(n^{(3-\tau)/2}/\log(n))$ with high probability for all vertices in $G'$. Since $G'$ consists of $O(n^{(3-\tau)/2}\log(n)^{\tau-1})$ vertices, $D_{i,G'}=O(n^{(3-\tau)/2}\log(n)^{\tau-1})$ as well. This means that for $k<\log^{\frac{1}{3}}(n)$, Steps 8-11 are able to find a connected subgraph on $k$ vertices with high probability.


We now compute the probability that $S_j$ is disjoint with the previous $j-1$ constructed sets. The number of vertices in sets $S_{j-1},\dots, S_1$ is bounded by $(j-1)k$. The probability that the first vertex does not overlap with the previous sets is therefore bounded by $1-(j-1)k/\abs{V'}$, since that vertex is chosen uniformly at random. The second vertex is chosen in a size-biased manner, since it is chosen by following a random edge. The probability that vertex $i$ is added can therefore be bounded as
\begin{equation}
\Prob{\text{vertex $i$ is added}}=\frac{D_{i,G'}}{{2|E'|}}\leq \frac{M\log^{\tau-1}(n)}{\abs{V'}}
\end{equation}
for some constant $M>0$ by the conditions on the degrees. Therefore, the probability that $S_j$ does not overlap with one of the previously chosen (at most $(j-1)k$) vertices can be bounded from below by 
\begin{equation}
\Prob{S_j\text{ does not overlap with previous sets}}\geq \left(1-\frac{k(j-1)}{\abs{V'}}\right)\left(1-\frac{Mk(j-1)\log^{\tau-1}(n)}{\abs{V'}}\right)^{k-1}.
\end{equation}
Thus, the probability that all $j$ sets do not overlap can be bounded as
\begin{equation}
\Prob{S_j \cap S_{j-1}\dots \cap S_1=\emptyset}\geq \left(1-\frac{Mk(j-1)\log^{\tau-1}(n)}{\abs{V'}}\right)^{(j-1)k},
\end{equation}
which tends to one when $jk=o(n^{(3-\tau)/4})$. Let $s_{\text{dis}}$ denote the number of disjoint sets out of the $s$ sets constructed in Algorithm~\ref{alg:motif2}. Then, when $s=\Omega(n^\alpha)$ for some $\alpha>0$, $s_{\text{dis}}>n^{\beta}$ for some $\beta>0$ with high probability, because $k<\log^{1/3}(n)$.

The probability that $H$ is present as an induced subgraph is bounded similarly as in Theorem~\ref{thm:graphletfind}. We already know that $k-1$ edges are present. For all other $E-(k-1)$ edges of $H$, and all ${k\choose 2}-E$ edges that are not present in $H$, we can again use~\eqref{eq:pijub} and~\eqref{eq:pijlb} to bound on the probability of edges being present or not being present between vertices in $V'$. Therefore, we can bound the probability that $H$ is not found similarly to~\eqref{eq:pHbound} as
\begin{equation*}
\begin{aligned}
\P\left( H\text{ not in the partitions} \right)
&\leq \P\left( H\text{ not in the disjoint partitions} \right)\\
& \le \exp\left( -p_-^{k^2}(1-p_+)^{k^2}  s_{\text{dis}} \right).
\end{aligned}
\end{equation*}
Because $s_{\text{dis}}>n^\beta$ for some $\beta>0$, this term tends to zero exponentially. 
The running time of the partial search can be bounded similarly to~\eqref{eq:psearch} as
\begin{equation}
s{k\choose 2}\leq sk^2=O(n k),
\end{equation}
where we used that $s\leq n/k$.
\end{proof}

\section{Experimental results}\label{sec:exp}
Figure~\ref{fig:succ1} shows the success rate of Algorithm~\ref{alg:motif}, defined as the fraction of times Algorithm~\ref{alg:motif} succeeds in finding a cycle of size $k$ in an inhomogeneous random graph on $10^7$ vertices. Even though for large $n$ Algorithm~\ref{alg:motif} should find an instance of a cycle of size $k$ in step 7 of the algorithm with high probability, we see that Algorithm~\ref{alg:motif} never succeeds in finding one of size 7. The success rate of Algorithm~\ref{alg:motif} on smaller cycles is also far away from 1, because of the finite size effects discussed before.

\begin{figure}[tb]
	\centering
	\includegraphics[width=0.45\textwidth]{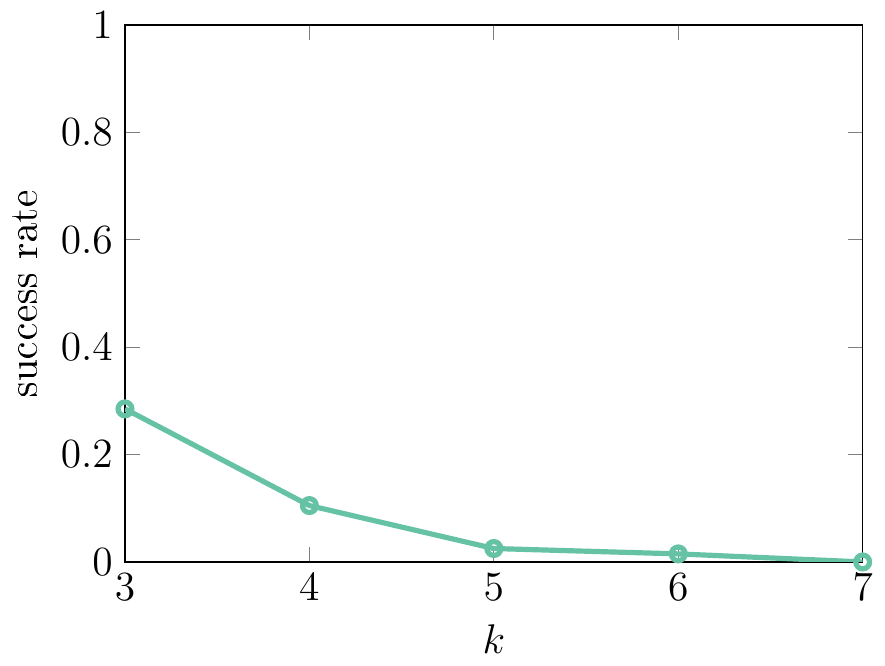}
	\caption{The fraction of times step 7 in Algorithm~\ref{alg:motif} succeeds to find a cycle of length $k$ in an inhomogeneous random graph with $N=10^7$, $\tau=2.5$, averaged over 500 network samples with $f_1=1/\log(n)$ and $f_2=0.9$.}
	\label{fig:succ1}
\end{figure}

Figure~\ref{fig:succ2} also plots the fraction of times Algorithm~\ref{alg:motif2} succeeds to find a cycle. We set the parameter $s=10000$ so that the algorithm fails if the algorithm does not succeed to detect motif $H$ after executing step 13 of Algorithm~\ref{alg:motif2} 10000 times. {Because $s$ gives the number of attempts to find $H$, increasing $s$ may increase the success probability of Algorithm~\ref{alg:motif2} at the cost of a higher running time. However, in Figure~\ref{fig:checks}, for small values of $k$, the mean number of times Step 13 is executed when the algorithm succeeds is much lower than 10000, so that increasing $s$ in this experiment probably only has a small effect on the success probability.} 

Algorithm~\ref{alg:motif2} with $f_1=1/\log(n)$ and $f_2=0.9$ in line with Theorem~\ref{thm:alg2} outperforms Algorithm~\ref{alg:motif}. Figure~\ref{fig:checks} also shows that the number of attempts needed to detect a cycle of length $k$ is small for $k\leq 6$. For larger values of $k$ the number of attempts increases. This can again be ascribed to the finite size effects that cause the set $V'$ to be small, so that large motifs may not be present on vertices in set $V'$. 

We also plot the success probability when using different values of the functions $f_1$ and $f_2$, outside the window where Theorem~\ref{thm:alg2} holds. When $f_2=\infty$, all vertices of degree at least $\sqrt{f_1\mu n}$ are included in $V'$, as in~\cite{friedrich2015}. In this setting, the success probability of the algorithm decreases. This is because the set $V'$ now contains many high degree vertices that are much more likely to form clique motifs than cycles or other connected motifs on $k$ vertices. This makes $f_1=1/\log(n), f_2=\infty$ a very efficient setting for detecting clique motifs~\cite{friedrich2015}. For the cycle motif however, Figure~\ref{fig:checks} shows that more checks are needed before a cycle is detected, and in some cases the cycle is not detected at all. 

The setting $f_1=0$ and $f_2=\infty$ is equivalent to including all vertices in $|V'|$. This is also less efficient than the setting of Theorem~\ref{thm:alg2}, as Figure~\ref{fig:succ2} shows. In this situation, the number of attempts needed to find a cycle of length $k$ is larger than for Algorithm~\ref{alg:motif2} for $k\leq 6$.

\begin{figure}[htb]
	\centering
	\begin{subfigure}{0.45\linewidth}
		\centering
		\includegraphics[width=\textwidth]{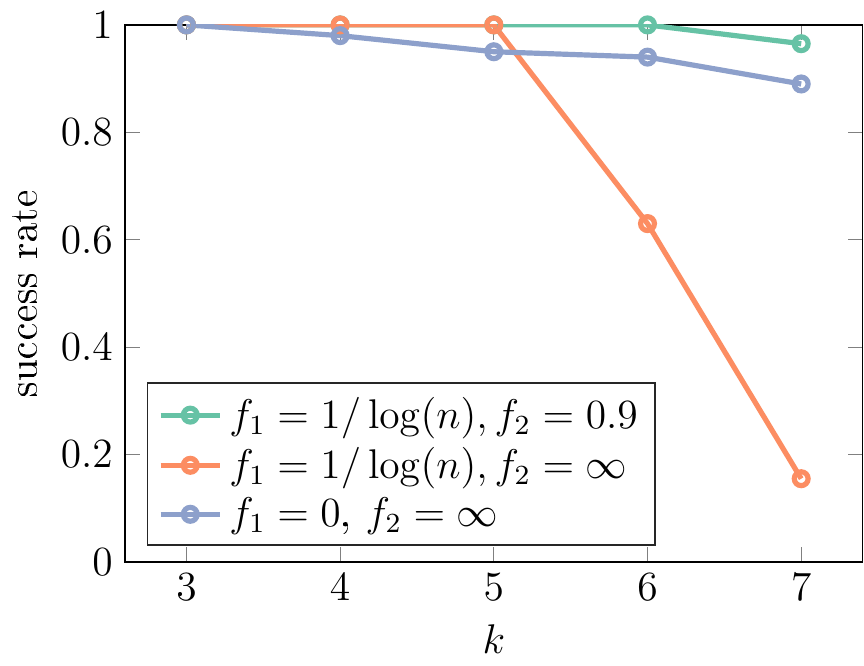}
		\caption{Success probability}
		\label{fig:succ2}
	\end{subfigure}
\begin{subfigure}{0.45\linewidth}
	\centering
	\includegraphics[width=\textwidth]{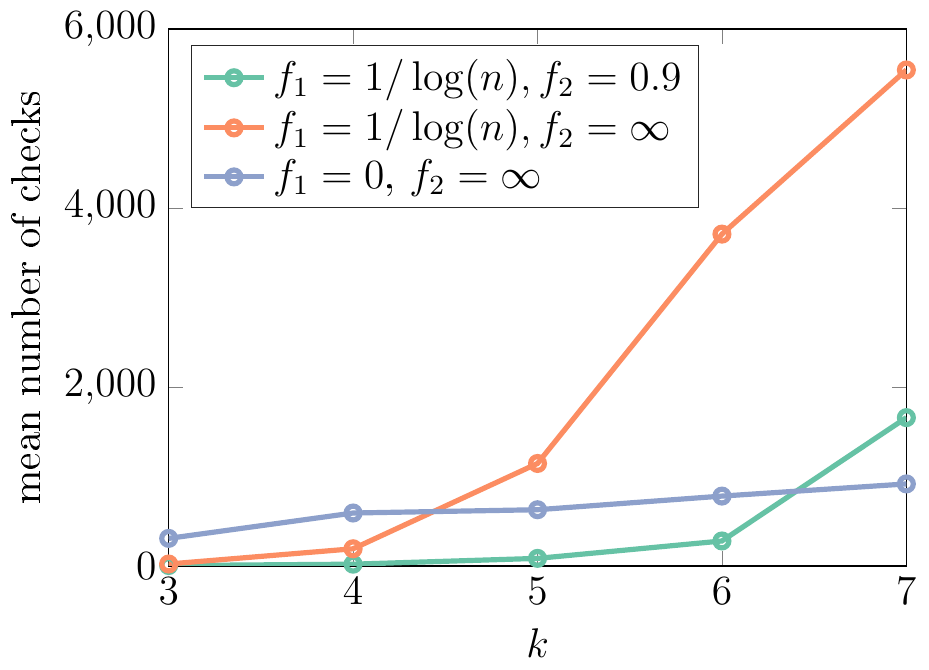}
	\caption{Mean number of checks}
	\label{fig:checks}
\end{subfigure}
\caption{Results of Algorithm~\ref{alg:motif2} on an inhomogeneous random graph with $N=10^7$, $\tau=2.5$ for detecting cycles of length $k$, using $s=10000$. The values are averaged over 500 generated networks.}
\end{figure}

Most results in this section are about finding cycles, since cycles easily scale in $k$ (larger cycles). However, Algorithm~\ref{alg:motif2} with parameters as in Theorem~\ref{thm:alg2} also helps to find other subgraphs than cycles. Table~\ref{tab:succhouse} presents the probability that the algorithm detects the subgraph of Figure~\ref{fig:subghouse}. Indeed, the probability that Algorithm~\ref{alg:motif2} with parameters chosen as in Theorem~\ref{thm:alg2} succeeds is much larger than the probability that the algorithm succeeds on the whole data set.

\begin{table}[htbp]
	\centering
	\begin{tabular}{llll}
		\toprule
		& $f_1=1/\log(n), f_2=0.9$ & \multicolumn{1}{l}{$f_1=1/\log(n), f_2=\infty$} & \multicolumn{1}{l}{$f_1=0, f_2=\infty$} \\
		\midrule
		\textbf{Success rate} & 0.47  & 0     & 0 \\
		\bottomrule
	\end{tabular}%
\caption{Success rate of Algorithm~\ref{alg:motif2} on the subgraph of Figure~\ref{fig:subghouse} for $N=10^6$, $\tau=2.5$, $s=10000$ over 500 generated networks.}
	\label{tab:succhouse}%
\end{table}%

\begin{figure}[tb]
	\centering
	\definecolor{mycolor1}{HTML}{66c2a5}
	\tikzstyle{every node}=[circle,fill=mycolor1,minimum size=8pt,inner sep=0pt,draw=black!80]
	\begin{tikzpicture}
	\centering
	\tikzstyle{edge} = [draw,thick,-]
	\node[] (a) at (0,0) {};
	\node[] (c) at (0,1) {};
	\node[] (e) at (0.5,1.5) {};
	\node[] (b) at (1,0) {};
	\node[] (d) at (1,1) {};
	\draw[edge] (a)--(c);
	\draw[edge] (a)--(b);
	\draw[edge] (e)--(c);
	\draw[edge] (d)--(e);
	\draw[edge] (b)--(d);
	\draw[edge] (d)--(c);
	\end{tikzpicture}
	\caption{The subgraph corresponding to the results in Table~\ref{tab:succhouse}.}
	\label{fig:subghouse}
\end{figure}
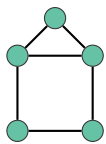

\subsection{Real network data}
We now check Algorithm~\ref{alg:motif2} on four real-world networks with power-law degrees: a Wikipedia communication network~\cite{snap}, the Gowalla social network~\cite{snap}, the Baidu online encyclopedia~\cite{niu2011} and the Internet on the autonomous systems level~\cite{snap}. {Table~\ref{tab:data} presents several statistics of these scale-free data sets.} Figure~\ref{fig:succdata} shows the fraction of runs where Algorithm~\ref{alg:motif2} finds a cycle as an induced subgraph. We see that for the Wikipedia social network in Figure~\ref{fig:succwiki}, Algorithm~\ref{alg:motif2} with $f_1=1/\log(n)$ and $f_2=0.9$ in line with Theorem~\ref{thm:alg2} is more efficient than looking for cycles among all vertices in the network ($f_1=0,f_2=\infty$). For the Baidu online encyclopedia in Figure~\ref{fig:succbaidu} however, we see that Algorithm~\ref{alg:motif2} with $f_1=1/\log(n)$ and $f_2=0.9$ performs much worse than looking for cycles among all possible vertices. In the other two network data sets in Figures~\ref{fig:succgowalla} and~\ref{fig:succasskit} the performance on the reduced vertex set and the original vertex set is almost the same. Figure~\ref{fig:meandata} shows that in general, Algorithm~\ref{alg:motif2} with settings as in Theorem~\ref{thm:alg2} indeed seems to finish in fewer steps than when using the full vertex set. However, as Figure~\ref{fig:meanbaidu} shows, for larger values of $k$ the algorithm fails almost always.

\begin{table}[tbp]
	\centering
	\begin{tabular}{lrrr}
		\toprule
		& $n$ & $E$ & $\tau$ \\
		\midrule
		\textbf{Wikipedia} & 2,394,385 & 5,021,410 & 2.46 \\
		\textbf{Gowalla} & 196,591 & 950,327 & 2.65  \\
		\textbf{Baidu} & 2,141,300 & 17,794,839 & 2.29  \\
		\textbf{AS-Skitter} & 1,696,415 & 11,095,298 & 2.35   \\
		\bottomrule \\
	\end{tabular}%
\caption{Statistics of the data sets: the number of vertices $n$, the number of edges $E$, and the power-law exponent $\tau$ fitted by the method of~\cite{clauset2009}.}
	\label{tab:data}%
\end{table}%

These results show that while Algorithm~\ref{alg:motif2} with $f_1,f_2$ in line with Theorem~\ref{thm:alg2} is efficient on inhomogeneous random graphs, it may not always be efficient on real-world data sets. This is not surprising, because there is no reason why in real-world data the vertices of degrees proportional to $\sqrt{n}$ should behave like an Erd\H{o}s-R\'enyi random graph, like in the inhomogeneous random graph. Thus, in terms of subgraphs, the inhomogeneous random graph and real-world network data differ significantly.

We therefore investigate whether selecting vertices with degrees in $I_n=[(\mu n)^\gamma/\log(n),(\mu n)^\gamma]$ for some other value of $\gamma$ in Algorithm~\ref{alg:motif2} leads to a better performance. Figure~\ref{fig:succdata} and~\ref{fig:meandata} show for every data set one particular value of $\gamma$ that works well. For the Gowalla, Wikipedia and Autonomous systems network, this leads to a faster algorithm to detect cycles. In these examples, the success probability of the algorithm is similar to the success probability on the full data set, but Figure~\ref{fig:meandata} shows that it finds the cycle much faster than the algorithm on the full data set. Only for the Baidu network other values of $\gamma$ do not improve upon randomly selecting from all vertices. This indicates that for most networks, cycles do appear mostly on degrees with specific orders of magnitude, making it possible to sample these cycles faster. Unfortunately, these orders of magnitude may be different for different networks. Across all four networks, the best value of $\gamma$ seems to be smaller than the value of 0.5 that is optimal for the inhomogeneous random graph.

{In these experiments, we tested the values $\gamma=0.1,0.2,0.3,0.4,0.5$, and we present in Figure~\ref{fig:succdata} the values of $\gamma$ that worked best for each data set. However, it would be useful to be able to select the best value of $\gamma$ without trying several values at first. For example, it may be possible to relate $\gamma$ to the degree-exponent $\tau$, or to a specific quantile of the degree sequence. Finding efficient methods to estimate $\gamma$ from a data set is an interesting question for further work. }

\begin{figure}[tb]
	\centering
	\begin{subfigure}{0.45\linewidth}
		\centering
		\includegraphics[width=\textwidth]{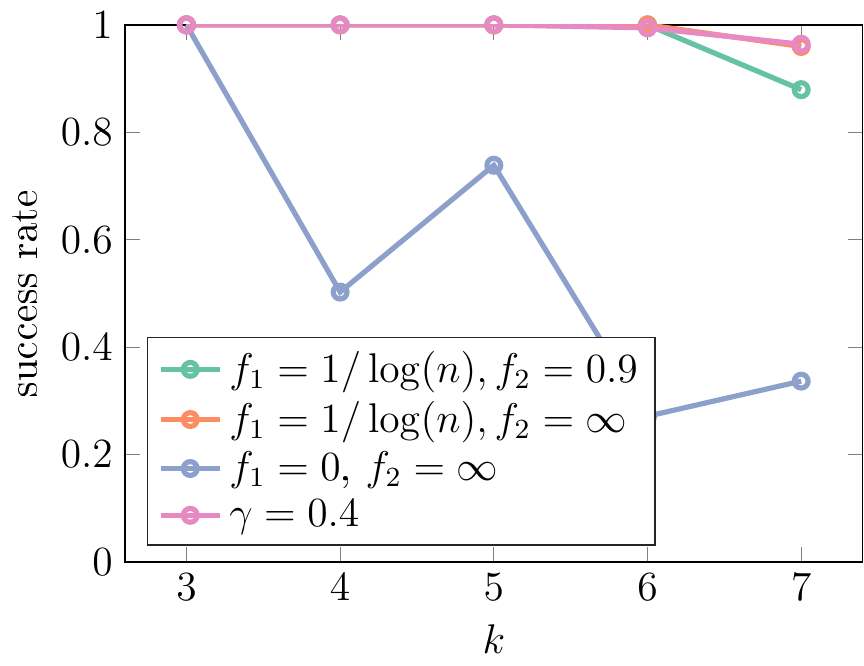}
		\caption{Wikipedia}
		\label{fig:succwiki}
	\end{subfigure}
\hspace{0.2cm}
	\begin{subfigure}{0.45\linewidth}
		\centering
		\includegraphics[width=\textwidth]{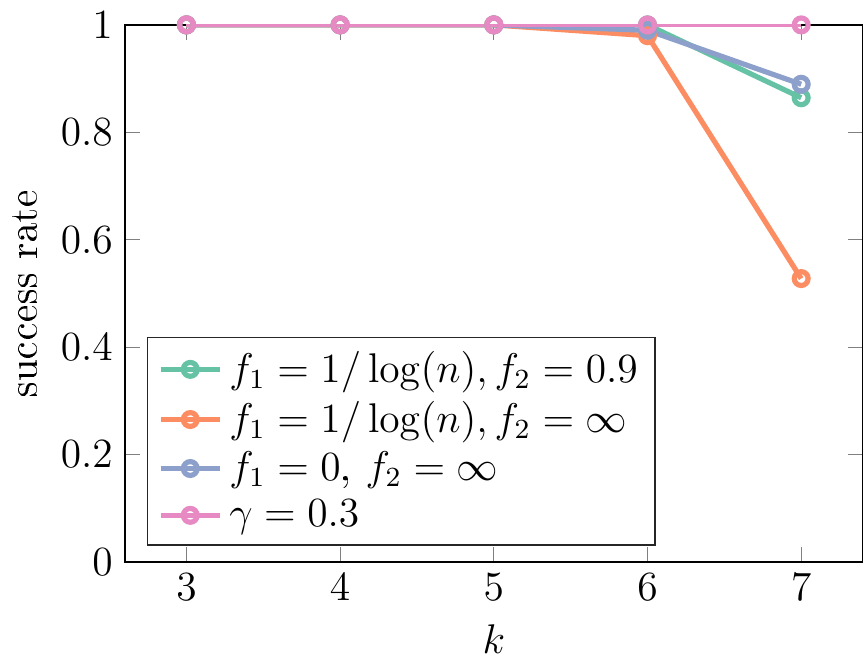}
		\caption{Gowalla}
		\label{fig:succgowalla}
	\end{subfigure}

\begin{subfigure}{0.45\linewidth}
	\centering
	\includegraphics[width=\textwidth]{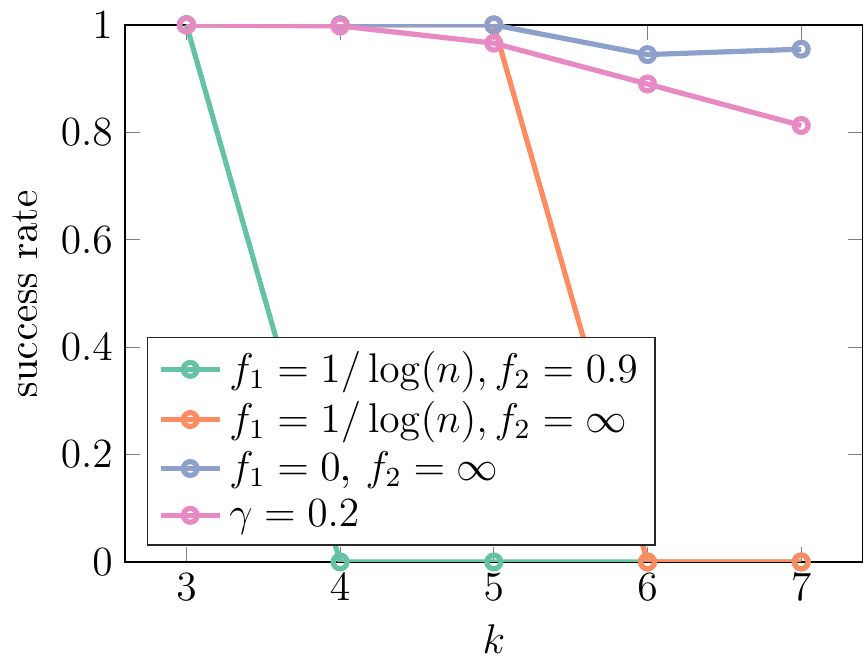}
	\caption{Baidu}
	\label{fig:succbaidu}
\end{subfigure}
\hspace{0.2cm}
\begin{subfigure}{0.45\linewidth}
	\centering
	\includegraphics[width=\textwidth]{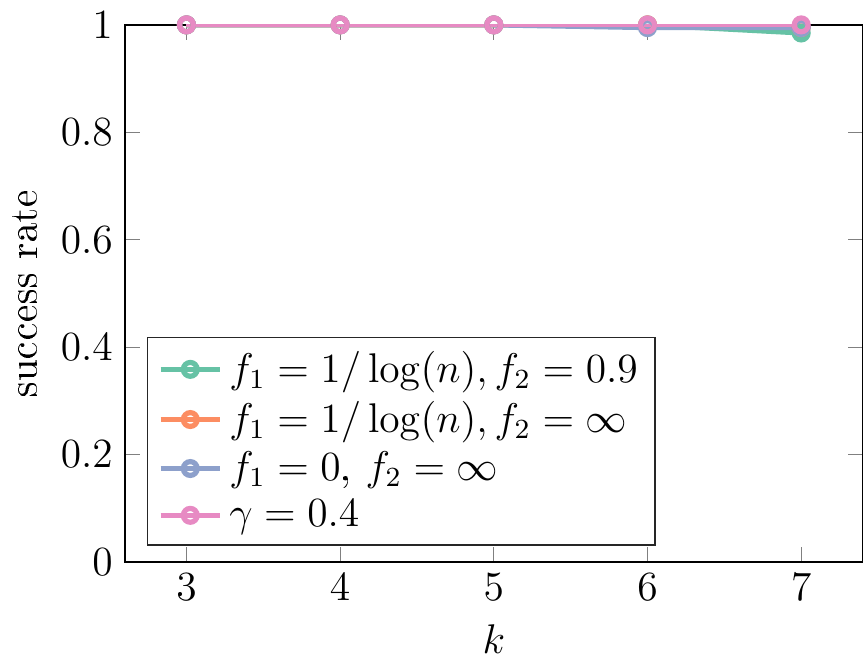}
	\caption{AS-Skitter}
	\label{fig:succasskit}
\end{subfigure}
	\caption{The fraction of times Algorithm~\ref{alg:motif2} succeeds to find a cycle on four large network data sets for detecting cycles of length $k$, using $s=10000$. The pink line uses Algorithm~\ref{alg:motif2} on vertices of degrees in $I_n=[(\mu n)^\gamma/\log(n),(\mu n)^\gamma]$. The values are averaged over 500 runs of Algorithm~\ref{alg:motif2}.}
	\label{fig:succdata}
	\end{figure}

\begin{figure}[htb]
	\centering
	\begin{subfigure}{0.45\linewidth}
		\centering
		\includegraphics[width=\textwidth]{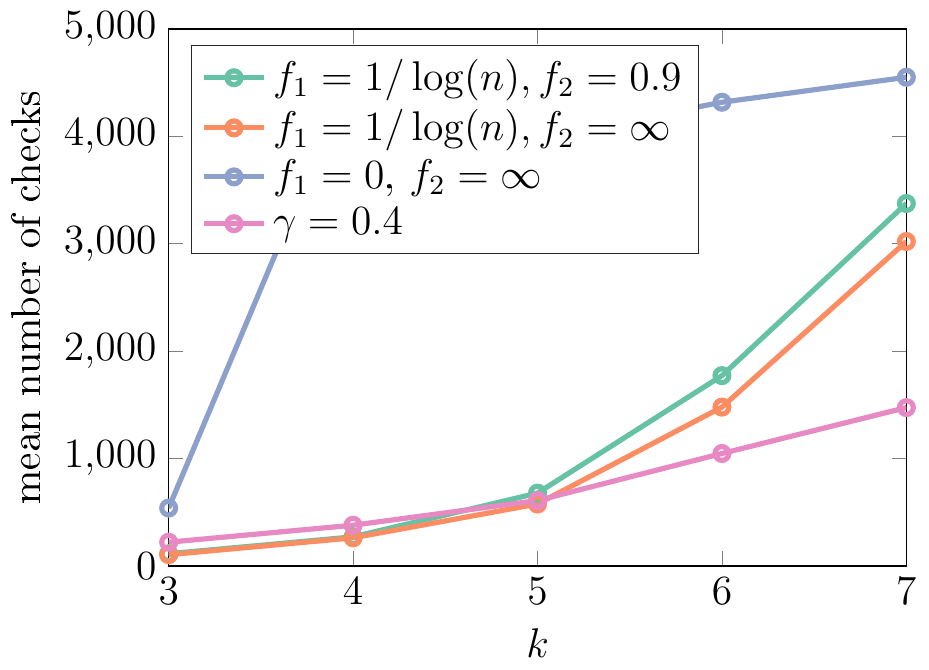}
		\caption{Wikipedia}
		\label{fig:meanwiki}
	\end{subfigure}
	\hspace{0.2cm}
	\begin{subfigure}{0.45\linewidth}
		\centering
		\includegraphics[width=\textwidth]{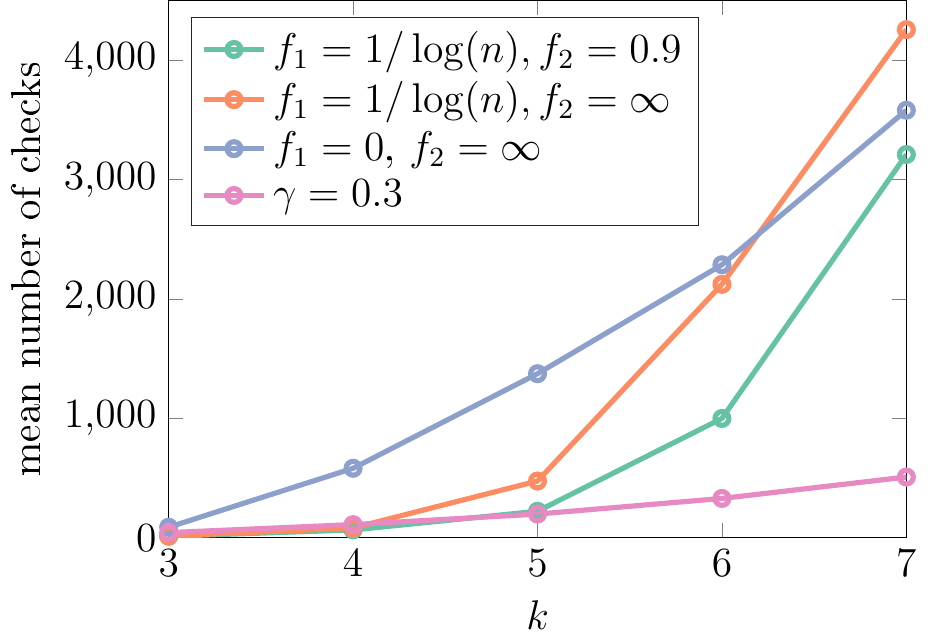}
		\caption{Gowalla}
		\label{fig:meangowalla}
	\end{subfigure}
	
	\begin{subfigure}{0.45\linewidth}
		\centering
		\includegraphics[width=\textwidth]{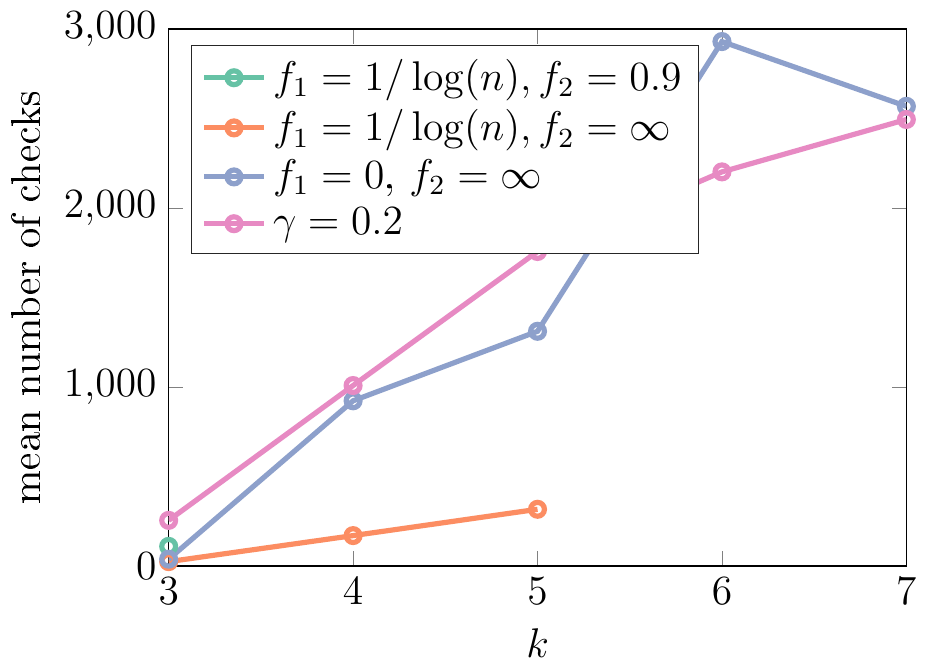}
		\caption{Baidu}
		\label{fig:meanbaidu}
	\end{subfigure}
	\hspace{0.2cm}
	\begin{subfigure}{0.45\linewidth}
		\centering
		\includegraphics[width=\textwidth]{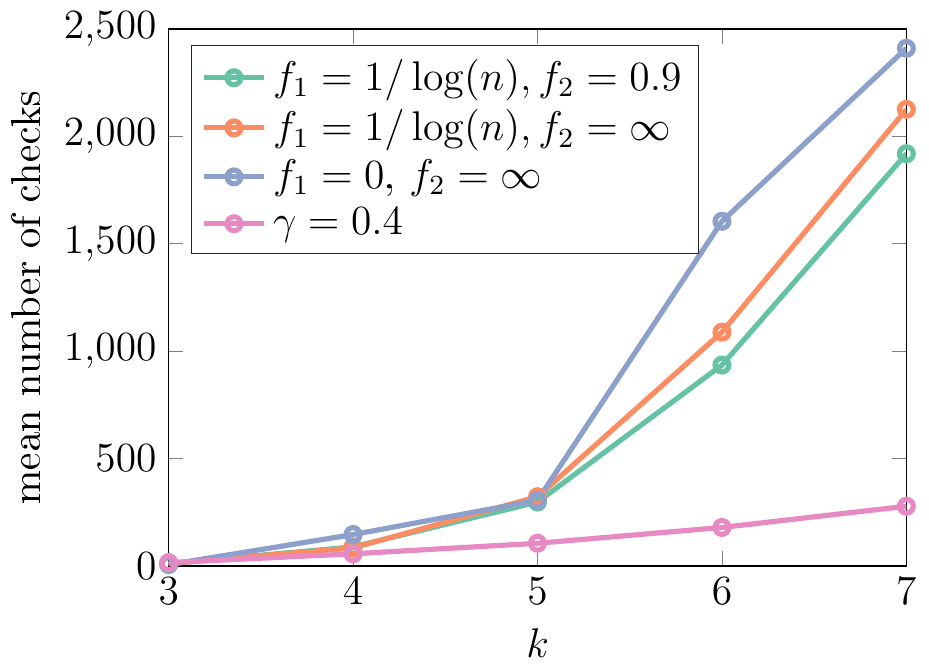}
		\caption{AS-Skitter}
		\label{fig:meanasskit}
	\end{subfigure}
	\caption{The number of times step 12 of Algorithm~\ref{alg:motif2} is invoked when the algorithm does not fail on four large network data sets for detecting cycles of length $k$, using $s=10000$. The pink line uses Algorithm~\ref{alg:motif2} on vertices of degrees in $I_n=[(\mu n)^\gamma/\log(n),(\mu n)^\gamma]$. The values are averaged over 500 runs of Algorithm~\ref{alg:motif2}.}
	\label{fig:meandata}
\end{figure}

\section{Conclusion}
We presented an algorithm which solves the induced subgraph problem on inhomogeneous random graphs with infinite variance power-law degrees in time $O(n\me^{k^4})$ with high probability as $n$ grows large.
This algorithm is based on the observation that for fixed $k$, any subgraph is present on $k$ vertices with degrees slightly smaller than $\sqrt{\mu n}$ with positive probability. Therefore, the algorithm first selects vertices with those degrees, and then uses a random search method to look for the induced subgraph on those vertices. 

We show that this algorithm performs well on simulations of inhomogeneous random graphs. Its performance on real-world data sets varies for different data sets. This indicates that the degrees that contain the most induced subgraphs of size $k$ in real-world networks may not be close to $\sqrt{n}$. We then show that on these data sets, it may be more efficient to find induced subgraphs on degrees proportional to $n^{\gamma}$ for some other value of $\gamma$. The value of $\gamma$ may be different for different networks.

Our algorithm exploits that induced subgraphs are likely formed among $\sqrt{\mu n}$-degree vertices. However, certain subgraphs may occur more frequently on vertices of other degrees~\cite{hofstad2017d}. For example, star-shaped subgraphs on $k$ vertices appear more often on one vertex with degree much higher than $\sqrt{\mu n}$ corresponding to the middle vertex of the star, and $k-1$ lower-degree vertices corresponding to the leafs of the star~\cite{hofstad2017d}. An interesting open question is whether there exist better degree-selection steps for specific subgraphs than the one used in Algorithms~\ref{alg:motif} and~\ref{alg:motif2}.

\paragraph{Acknowledgements.}
The work of JvL and CS was supported by NWO TOP grant 613.001.451. The work of JvL was further supported by the NWO Gravitation Networks grant 024.002.003, an NWO TOP-GO grant and by an ERC Starting Grant.

	\bibliographystyle{splncs_srt}
	\bibliography{references}

\begin{thebibliography}{10}

\bibitem{albert1999}
Albert, R., Jeong, H., Barab{\'a}si, A.L.:
\newblock Internet: {D}iameter of the world-wide web.
\newblock Nature \textbf{401}(6749) (1999)  130--131

\bibitem{boguna2003}
Bogu\~n\'a, M., Pastor-Satorras, R.:
\newblock Class of correlated random networks with hidden variables.
\newblock Phys. Rev. E \textbf{68} (2003)  036112

\bibitem{bollobas2007}
Bollob{\'a}s, B., Janson, S., Riordan, O.:
\newblock The phase transition in inhomogeneous random graphs.
\newblock Random Structures \& Algorithms \textbf{31}(1) (2007)  3--122

\bibitem{brach2016}
Brach, P., Cygan, M., {\L}acki, J., Sankowski, P.:
\newblock Algorithmic complexity of power law networks.
\newblock In: Proceedings of the Twenty-seventh Annual ACM-SIAM Symposium on
  Discrete Algorithms. SODA '16, Philadelphia, PA, USA, Society for Industrial
  and Applied Mathematics (2016)  1306--1325

\bibitem{britton2006}
Britton, T., Deijfen, M., Martin-L\"{o}f, A.:
\newblock Generating simple random graphs with prescribed degree distribution.
\newblock J. Stat. Phys. \textbf{124}(6) (2006)  1377--1397

\bibitem{chung2002}
Chung, F., Lu, L.:
\newblock The average distances in random graphs with given expected degrees.
\newblock Proc. Natl. Acad. Sci. USA \textbf{{\bf 99}}(25) (2002)  15879--15882
  (electronic)

\bibitem{clauset2009}
Clauset, A., Shalizi, C.R., Newman, M.E.J.:
\newblock Power-law distributions in empirical data.
\newblock SIAM Rev. \textbf{51}(4) (2009)  661--703

\bibitem{faloutsos1999}
Faloutsos, M., Faloutsos, P., Faloutsos, C.:
\newblock On power-law relationships of the internet topology.
\newblock In: ACM SIGCOMM Computer Communication Review. Volume~29., ACM (1999)
   251--262

\bibitem{fountoulakis2014}
Fountoulakis, N., Friedrich, T., Hermelin, D.:
\newblock On the average-case complexity of parameterized clique.
\newblock \textup{arXiv:}1410.6400v1 (2014)

\bibitem{fountoulakis2015}
Fountoulakis, N., Friedrich, T., Hermelin, D.:
\newblock On the average-case complexity of parameterized clique.
\newblock Theoretical Computer Science \textbf{576} (apr 2015)  18--29

\bibitem{friedrich2015}
Friedrich, T., Krohmer, A.:
\newblock Cliques in hyperbolic random graphs.
\newblock In: INFOCOM proceedings 2015, IEEE (2015)  1544--1552

\bibitem{friedrich2015a}
Friedrich, T., Krohmer, A.:
\newblock Parameterized clique on inhomogeneous random graphs.
\newblock Discrete Applied Mathematics \textbf{184} (mar 2015)  130--138

\bibitem{garey2011}
Garey, M.R., Johnson, D.S., Garey, M.R.:
\newblock Computers and Intractability: A Guide to the Theory of
  NP-Completeness.
\newblock W H FREEMAN \& CO (1979)

\bibitem{grochow2007}
Grochow, J.A., Kellis, M.:
\newblock Network motif discovery using subgraph enumeration and
  symmetry-breaking.
\newblock In: In RECOMB. (2007)  92--106

\bibitem{heydari2017}
Heydari, H., Taheri, S.M.:
\newblock Distributed maximal independent set on inhomogeneous random graphs.
\newblock In: 2017 2nd Conference on Swarm Intelligence and Evolutionary
  Computation ({CSIEC}), {IEEE} (mar 2017)

\bibitem{hofstad2017d}
{\swap{Hofstad}{van der }}, R., van Leeuwaarden, J.S.H., Stegehuis, C.:
\newblock Optimal subgraph structures in scale-free networks.
\newblock \textup{arXiv:1709.03466} (2017)

\bibitem{janson2010}
Janson, S., {\L}uczak, T., Norros, I.:
\newblock Large cliques in a power-law random graph.
\newblock Journal of Applied Probability \textbf{47}(04) (dec 2010)  1124--1135

\bibitem{jeong2000}
Jeong, H., Tombor, B., Albert, R., Oltvai, Z.N., Barab{\'a}si, A.L.:
\newblock The large-scale organization of metabolic networks.
\newblock Nature \textbf{407}(6804) (2000)  651--654

\bibitem{karp1972}
Karp, R.M.:
\newblock Reducibility among combinatorial problems.
\newblock In: Complexity of computer computations.
\newblock Springer (1972)  85--103

\bibitem{kashtan2004}
Kashtan, N., Itzkovitz, S., Milo, R., Alon, U.:
\newblock Efficient sampling algorithm for estimating subgraph concentrations
  and detecting network motifs.
\newblock Bioinformatics \textbf{20}(11) (2004)  1746--1758

\bibitem{snap}
Leskovec, J., Krevl, A.:
\newblock {SNAP Datasets}: {Stanford} large network dataset collection.
\newblock \url{http://snap.stanford.edu/data} (2014) Date of access:
  14/03/2017.

\bibitem{niu2011}
Niu, X., Sun, X., Wang, H., Rong, S., Qi, G., Yu, Y.:
\newblock Zhishi.me - weaving chinese linking open data.
\newblock In: The Semantic Web {\textendash} {ISWC} 2011.
\newblock Springer Nature (2011)  205--220

\bibitem{norros2006}
Norros, I., Reittu, H.:
\newblock On a conditionally poissonian graph process.
\newblock Adv. Appl. Probab. \textbf{38}(01) (2006)  59--75

\bibitem{omidi2009}
Omidi, S., Schreiber, F., Masoudi-Nejad, A.:
\newblock {MODA}: An efficient algorithm for network motif discovery in
  biological networks.
\newblock Genes {\&} Genetic Systems \textbf{84}(5) (2009)  385--395

\bibitem{park2004}
Park, J., Newman, M.E.J.:
\newblock Statistical mechanics of networks.
\newblock Phys. Rev. E \textbf{70} (2004)  066117

\bibitem{schreiber2005}
Schreiber, F., Schwobbermeyer, H.:
\newblock {MAVisto}: a tool for the exploration of network motifs.
\newblock Bioinformatics \textbf{21}(17) (jul 2005)  3572--3574

\bibitem{vazquez2002}
V{\'a}zquez, A., Pastor-Satorras, R., Vespignani, A.:
\newblock Large-scale topological and dynamical properties of the internet.
\newblock Phys. Rev. E \textbf{65} (2002)  066130

\bibitem{williams2015}
Williams, V.V., Wang, J.R., Williams, R., Yu, H.:
\newblock Finding four-node subgraphs in triangle time.
\newblock In: Proceedings of the Twenty-sixth Annual ACM-SIAM Symposium on
  Discrete Algorithms. SODA '15, Philadelphia, PA, USA, Society for Industrial
  and Applied Mathematics (2015)  1671--1680

\end{thebibliography}
	
\end{document}